\documentclass[12pt]{article}
\usepackage[utf8]{inputenc}

\usepackage{bbm}
\usepackage{amssymb,amsmath,amscd,amsthm}
\usepackage{url}
\usepackage{xspace}
\usepackage{enumerate}

\usepackage[dvipsnames,usenames]{xcolor}
\usepackage{tikz}
\usetikzlibrary{fit}
\usetikzlibrary{shapes.geometric}
\usetikzlibrary{backgrounds}
\usetikzlibrary{patterns}
\usetikzlibrary{calc}
\usetikzlibrary{graphs}

\colorlet{lred}{red!10}
\colorlet{dgreen}{green!80!black!80}

\makeatletter
\def\@gifnextchar#1#2#3{\let\@tempe#1\def\@tempa{#2}\def\@tempb{#3}%
  \futurelet\@tempc\@gifnch}

\def\@gifnch{\ifx\@tempc\@sptoken\let\@tempd\@tempb%
  \else\ifx\@tempc\@tempe\let\@tempd\@tempa\else\let\@tempd\@tempb\fi\fi\@tempd}

\def\SK@set#1{\left\{#1\right\}}
\def\SK@@set#1#2{\{#1\,:\,
    \begin{array}{@{}l@{}}#2\end{array}
\}}
\def\SK@mset#1{\left\{\!\!\left\{#1\right\}\!\!\right\}}
\def\SK@@mset#1#2{\{\!\!\{#1\,:\,
    \begin{array}{@{}l@{}}#2\end{array}
\}\!\!\}}
\def\BIG@set#1{\Big\{#1\Big\}}
\def\BIG@@set#1#2{\Big\{#1\:\Big|\:
    \begin{array}{@{}l@{}}#2\end{array}
\Big\}}
\newcommand{\Set}[1]{\@gifnextchar\bgroup{\SK@@set{#1}}{\SK@set{#1}}}
\newcommand{\Mset}[1]{\@gifnextchar\bgroup{\SK@@mset{#1}}{\SK@mset{#1}}}
\newcommand{\Bigset}[1]{\@gifnextchar\bgroup{\BIG@@set{#1}}{\BIG@set{#1}}}
\makeatother

\theoremstyle{plain}
\newtheorem{theorem}{Theorem}[section]
\newtheorem{lemma}[theorem]{Lemma}

\newtheorem{corollary}[theorem]{Corollary}
\theoremstyle{definition} 

\newtheorem{remark}[theorem]{Remark}
\newtheorem{definition}[theorem]{Definition}

\newcommand{\refeq}[1]{(\ref{eq:#1})}

\newcommand{\Zset}{\mathbb{Z}}
\newcommand{\Rset}{\mathbb{R}}
\newcommand{\one}{\mathbbm{1}}

\newcommand{\of}[1]{\left( #1 \right)}
\newcommand{\function}[2]{:#1 \rightarrow #2}

\newcommand{\WL}[1]{\ensuremath{#1\text{-}\mathrm{WL}}\xspace}
\newcommand{\kWL}{\WL k}
\newcommand{\eqkwl}{\equiv_{\kWL}}
\newcommand{\eqkkwl}[1]{\equiv_{#1\text{-}\mathrm{WL}}}

\newcommand{\alg}[4]{\mathrm{WL}_{#1}^{#2}(#3,#4)}
\newcommand{\algstabbb}[1]{\mathrm{WL}_{#1}}
\newcommand{\algkstab}[2]{\mathrm{WL}_{k}(#1,#2)}
\newcommand{\algkstabb}[1]{\mathrm{WL}_{k}(#1)}
\newcommand{\algkstabbb}{\mathrm{WL}_{k}}
\newcommand{\barx}{{\bar x}}

\newcommand{\Sub}[2]{\mathrm{Sub}(#1,#2)}

\newcommand{\eSubb}[3]{\mathrm{Sub}(#1,#2,#3)}

\newcommand{\esub}[4]{\mathrm{sub}(#1,#2,#3,#4)}

\newcommand{\tw}[1]{\mathit{tw}(#1)}
\newcommand{\htw}[1]{\mathit{htw}(#1)}

\newcommand{\Pa}{{\mathcal P}}
\newcommand{\calS}{\mathcal{S}}

\newcommand{\ig}[1]{I(#1)}

\DeclareMathOperator{\cay}{Cay}
\newcommand{\bZ}{\mathbb{Z}}

\title{On the Weisfeiler-Leman Dimension\\ of Fractional Packing}
\author{V.~Arvind\thanks{The Institute of Mathematical Sciences (HBNI), Chennai, India.}, 
Frank Fuhlbrück\thanks{Institut für Informatik, Humboldt-Universität zu Berlin, Germany.}, 
Johannes Köbler${}^\dagger$, Oleg Verbitsky${}^\dagger$\,\thanks{Supported by DFG grant KO 1053/8--1. 
On leave from the IAPMM, Lviv, Ukraine.}}

\date{}

\begin{document}

\maketitle

\begin{abstract}
The $k$-dimensional Weisfeiler-Leman procedure (\kWL), which colors
$k$-tuples of vertices in rounds based on the neighborhood structure
in the graph, has proven to be immensely fruitful in the algorithmic
study of Graph Isomorphism. More generally, it is of fundamental
importance in understanding and exploiting symmetries in graphs in
various settings. Two graphs are $\kWL$-equivalent if the
$k$-dimensional Weisfeiler-Leman procedure produces the same final
coloring on both graphs. \WL1-equivalence is known as fractional
isomorphism of graphs, and the $\kWL$-equivalence relation becomes
finer as $k$ increases.

We investigate to what extent standard graph parameters are preserved
by $\kWL$-equivalence, focusing on fractional graph packing numbers.
The integral packing numbers are typically NP-hard to compute, and
we discuss applicability of $\kWL$-invariance for estimating
the integrality gap of the LP relaxation provided by their fractional counterparts.
\end{abstract}

\section{Introduction}\label{s:intro}

The $1$-dimensional version of the Weisfeiler-Leman procedure is the
classical \emph{color refinement} applied to an input graph $G$. Each
vertex of $G$ is initially colored by its degree. The procedure
refines the color of each vertex $x\in V(G)$ in rounds, using the
multiset of vertex colors in the neighborhood of $x$. In the
$2$-dimensional version~\cite{WLe68}, all vertex pairs $(x,y)\in
V(G)\times V(G)$ are classified by a similar procedure of coloring
them in rounds. The extension of this procedure to a classification of
all $k$-tuples of $G$ is due to Babai (see historical overview in
\cite{Babai16,CaiFI92}) and is known as the \emph{$k$-dimensional
  Weisfeiler-Leman procedure}, abbreviated as~\kWL.
Graphs $G$ and $H$ are said to be \emph{\kWL-equivalent} (denoted
$G\eqkwl H$) if they are indistinguishable by \kWL.

\paragraph{\textbf{The WL invariance of graph parameters.}}

Let $\mathcal{G}$ denote the set of all graphs. A \emph{graph parameter}
is a function $\pi$ defined on $\mathcal{G}$ such that $\pi(G)=\pi(H)$
whenever $G$ and $H$ are isomorphic.

We say that $\pi$ is \emph{\kWL-invariant} if the equality
$\pi(G)=\pi(H)$ is implied even by the weaker condition $G\eqkwl H$.

\begin{definition}\label{def-dim}
  The \emph{Weisfeiler-Leman (WL) dimension} of a graph parameter
  $\pi$ is the least positive integer $k$, if it exists, such that for
  any pairs of graphs $G$ and $H$ that are $\kWL$-indistinguishable we
  have $\pi(G)=\pi(H)$. If no such $k$ exists, we say that the WL
  dimension of $\pi$ is \emph{unbounded}.
\end{definition}

Knowing that a parameter $\pi$ has unbounded WL dimension is important
from a descriptive complexity perspective, because it implies that
$\pi$ cannot be computed by any algorithm expressible in fixed-point
logic with counting (FPC), which is a robust framework for study of
\emph{encoding-invariant (\emph{or} ``choiceless'') computations}; see
the survey~\cite{Dawar15}.

The focus of our paper is on graph parameters with \emph{bounded} WL
dimension. A subset $\mathcal{P}\subseteq \mathcal{G}$ is a
\emph{graph property} if the indicator function $\pi$ of $\mathcal{P}$
is a graph parameter. It is well-known in finite model theory \cite{CaiFI92} that
$\pi$ is \kWL-invariant if and only if $\mathcal{P}$ is definable in
the infinitary $(k+1)$-variable counting logic
$C^{k+1}_{\infty\omega}$. While minimizing the number of
variables is a recurring theme in descriptive complexity; see,
e.g.~\cite{ImmermanK89,Fortin19}, our interest in the study of
\kWL-invariance has an additional motivation: If we know that a graph
parameter $\pi$ is \kWL-invariant, this gives us information not only
about $\pi$ but also \emph{about \kWL}.

Indeed, \kWL-invariance admits the following interpretation.  We say
that a (not necessarily numerical) graph invariant $\pi_1$
\emph{subsumes} a graph invariant $\pi_2$ if $\pi_1(G)\ne\pi_1(H)$
whenever $\pi_2(G)\ne\pi_2(H)$. That is to say, whenever $\pi_2$
distinguishes between graphs $G$ and $H$, $\pi_1$ also does. Let
$\algkstabb G$ denote the graph invariant computed by $\kWL$ on input
$G$.  As easily seen, a parameter $\pi$ is \kWL-invariant if and only
if $\pi$ is subsumed by $\algkstabbb$.  Which graph parameters are
subsumed by $\algkstabbb$ is of interest even for dimensions $k=1$ and
$k=2$, in view of the importance of \WL1 (color refinement) and \WL2
(the original Weisfeiler-Leman algorithm) in isomorphism testing
\cite{Babai16,BabaiES80} and, more recently, also in other application
areas \cite{KriegeJM19,Morris+19}.  It is known, for example, that the
largest eigenvalue of the adjacency matrix has WL dimension 1 (see
\cite{ScheinermanU97}), and the whole spectrum of a graph has WL
dimension 2 (see \cite{DawarSZ17,Fuerer10}). Relatedly, Kiefer and
Neuen \cite{KieferN19} recently proved that $\algstabbb 2$ subsumes,
in a certain strong sense, the decomposition of a graph into
3-connected components.

\paragraph{\textbf{Fractional graph parameters.}}

In this paper, we mainly consider \emph{fractional} graph parameters.
Algorithmically, a well-known approach to tackling intractable
optimization problems is to consider an appropriate linear programming
(LP) relaxation. Many standard integer-valued graph parameters have
fractional real-valued analogues, obtained by LP-relaxation of the
corresponding 0-1 linear program; see, e.g., the monograph
\cite{ScheinermanU97}. The fractional counterpart of a graph parameter
$\pi$ is denoted by $\pi_f$. While $\pi$ is often be hard to compute,
the fractional parameter $\pi_f$ sometimes provides a good
polynomial-time computable approximation of~$\pi$.

The WL dimension of a natural fractional parameter $\pi_f$ is a priori
bounded, where \emph{natural} means that $\pi_f$ is determined by an
LP which is logically interpretable in terms of an input graph $G$. A
striking result of Anderson, Dawar, Holm \cite{AndersonDH15} says that
the optimum value of an interpretable LP is expressible in FPC.  It
follows from the known immersion of FPC into the finite-variable
infinitary counting logic
$C^\omega_{\infty\omega}=\bigcup_{k=2}^\infty C^k_{\infty\omega}$ (see
\cite{Otto-b}), that each such $\pi_f$ is \kWL-invariant for some
$k$. Although this general theorem is applicable to many graph
parameters of interest, it is not \emph{a priori} evident how to
extract an explicit value of $k$ from the proof of the theorem, and in
any case such a value of $k$ seems unlikely to be optimal.

We are interested in \emph{explicit} and, possibly, \emph{exact}
bounds for the WL dimension. A first question here would be to
pinpoint which fractional parameters $\pi_f$ are
$\WL1$-invariant. This natural question, using the concept of
fractional isomorphism \cite{ScheinermanU97}, can be recast as
follows: Which \emph{fractional} graph parameters are invariant under
\emph{fractional} isomorphisms? It appears that this question has not
received adequate attention in the literature. The only earlier result
we could find is the \WL1-invariance of the fractional domination
number $\gamma_f$ shown in the Ph.D.~thesis of
Rubalcaba~\cite{Rubalcaba-th}.

We show that the fractional matching number $\nu_f$ is also a
fractional parameter preserved by fractional isomorphism. Indeed, the
matching number is an instance of the \emph{$F$-packing number
  $\pi^F$} of a graph, corresponding to $F=K_2$.  Here and throughout,
we use the standard notation $K_n$ for the complete graphs, $P_n$ for
the path graphs, and $C_n$ for the cycle graph on $n$ vertices.  In
general, $\pi^F(G)$ is the maximum number of vertex-disjoint subgraphs
of $G$ that are isomorphic to the fixed pattern graph $F$. While
the matching number is computable in polynomial time, computing
$\pi^F$ is NP-hard whenever $F$ has a connected component with at
least 3 vertices \cite{KirkpatrickH83}, in particular, for
$F\in\{P_3,K_3\}$.  Note that $K_3$-packing is the optimization
version of the archetypal NP-complete problem Partition Into Triangles
\cite[GT11]{GareyJ79}.  We show that the fractional $P_3$-packing
number $\nu^{P_3}_f$, like $\nu_f=\pi_f^{K_2}$, is \WL1-invariant,
whereas the WL dimension of the fractional triangle packing is~2.

In fact, we present a general treatment of fractional $F$-packing
numbers~$\pi^F_f$.  We begin in Section \ref{s:LPs} with introducing a
concept of equivalence between two linear programs $L_1$ and $L_2$
ensuring that equivalent $L_1$ and $L_2$ have equal optimum values.
Next, in Section \ref{s:starter}, we consider the standard optimization
versions of Set Packing and Hitting Set \cite[SP4 and SP8]{GareyJ79},
two of Karp's 21 NP-complete problems \cite{Karp72}.
These two generic problems generalize $F$-Packing and Dominating Set respectively.  
Their fractional versions have thoroughly been studied in hypergraph theory \cite{Lovasz75,Fueredi88}.
We observe that the
LP relaxations of Set Packing (or Hitting Set) are equivalent whenever
the incidence graphs of the input set systems are
\WL1-equivalent. This general fact readily implies Rubalcaba's
result~\cite{Rubalcaba-th} on the \WL1-invariance of the fractional
domination number and also shows that, if the pattern graph $F$ has
$\ell$ vertices, then the fractional $F$-packing number $\pi^F_f$ is
\WL k-invariant for some $k<2\,\ell$.  This bound for $k$ comes from a
logical definition of the instance of Set Packing corresponding to
$F$-Packing in terms of an input graph $G$ (see Section
\ref{ss:excess}).  Though the bound is quite decent, it does not need
to be optimal.  We elaborate on a more precise bound, where we need to
use additional combinatorial arguments even in the case of the
fractional matching.

We treat the fractional matching separately for expository purposes in Section \ref{s:starter}.
The general $F$-packing is considered in Section \ref{sec:pack}, where
our main result, Theorem \ref{thm:htwpack}, includes the aforementioned cases of $F=K_3,P_3$.

The \emph{edge-disjoint} version of $F$-Packing is
another problem that has intensively been studied in combinatorics and
optimization. Since it is known to be NP-hard for any pattern $F$
containing a connected component with at least 3 edges \cite{DorT97},
fractional relaxations have received much attention in the literature
\cite{Dross16,HaxellR01,Yuster05,Yuster07}.
We show that our techniques work well also in this case.
In particular, the WL dimension of the fractional edge-disjoint triangle packing number
$\rho^{K_3}_f$ is~2 (Theorem~\ref{thm:rho_f}).

\paragraph{\textbf{Integrality gap via invariance ratio.}}

Furthermore, we discuss the \emph{approximate} invariance of
\emph{integral} graph parameters expressible by integer linear
programs.  As a first example, note that the \WL1-invariance of the
fractional matching number $\nu_f$ has two consequences. The first
follows from the known fact \cite[Theorem 2.1.3]{ScheinermanU97} that
$\nu_f(G)=\nu(G)$ if $G$ is bipartite. This equality implies that \emph{over bipartite}
graphs even the integral parameter $\nu$ is \WL1-invariant;
cf.~\cite{BlassGS02,AtseriasM13}.

Another consequence concerns all graphs and is based on 
Lov{\'{a}}sz's inequality \cite[Theorem 5.21]{Fueredi88}
\begin{equation}
  \label{eq:nunuf}
\nu_f(G)\le\frac32(\tau(G)+\nu(G))\le\frac32\,\nu(G)  
\end{equation}
where $\tau(G)$ is the domination number of a graph $G$. As $\nu_f$ is
\WL1-invariant, it follows that
\begin{equation}
  \label{eq:nunu}
\nu(G)/\nu(H)\le3/2  
\end{equation}
for any pair of nonempty \WL1-equivalent graphs $G$ and $H$. The bound
\refeq{nunu} is tight, as seen for the \WL1-equivalent graphs
$G=C_{6s}$ and $H=2s\,C_3$. Consequently, the relationship between 
$\nu_f$ and $\nu$ given by \refeq{nunuf} is also
tight. This simple example shows that the exact value $k$ of the WL
dimension of a fractional parameter $\pi_f$, and the discrepancy of
the integral parameter $\pi$ over {\WL k}-invariant graphs together
yield a lower bound for the precision of approximating $\pi$
by~$\pi_f$.

Specifically, recall that the maximum 
\[
\max_G \frac{\pi_f(G)}{\pi(G)},
\]
(respectively, $\max_G \pi(G)/\pi_f(G)$ for minimization problems) is
known as the \emph{integrality gap} of $\pi_f$. The integrality gap is
important for a computationally hard graph parameter $\pi$, as it
bounds how well the polynomial-time computable parameter $\pi_f$
approximates $\pi$. 

On the other hand, we define the \emph{{\WL k}-invariance ratio} for
the parameter $\pi$ as
\[
\max_{G,H}{\frac{\pi(G)}{\pi(H)}},
\]
where the quotient is maximized over all {\WL k}-equivalent graph
pairs $(G,H)$.  If $\pi$ is {\WL k}-invariant, then the {\WL
  k}-invariance ratio bounds the integrality gap from below. The
following question suggests itself: How tight is this lower bound? In
this regard, let us look at the fractional domination number~$\gamma_f$ again.

A general bound by Lov{\'{a}}sz \cite{Lovasz75} on the integrality gap
of the fractional covering number for hypergraphs implies for a graph $G$ that
$\gamma(G)\le(1+\ln(1+\Delta(G)))\,\gamma_f(G)$, where $\Delta(G)$
denotes the maximum vertex degree of $G$. It follows that
the integrality gap for the domination number is at most logarithmic.
More precisely,
\begin{equation}
  \label{eq:gamma-int-gap}
\frac{\gamma(G)}{\gamma_f(G)}\le 1+\ln n
\end{equation}
for a non-empty graph $G$ with $n$ vertices. This results in an
LP-based algorithm for approximation of $\gamma(G)$ within a
logarithmic factor, which is essentially optimal as $\gamma(G)$ is
inapproximable within a factor of $(1-\epsilon)\ln n$ unless
$\mathrm{NP}\subseteq\mathrm{DTIME}(n^{O(\log\log n)})$; see \cite{ChlebikC08}. Recall that $\gamma_f$ is \WL1-invariant.  Along
with \refeq{gamma-int-gap}, this implies that the \WL1-invariance
ratio of $\gamma$ is at most logarithmic. On the other hand, Chappell
et al.~\cite{ChappellGH17} have shown that the bound
\refeq{gamma-int-gap} is tight up to a constant factor. In Section
\ref{s:inv} we prove an $\Omega(\log n)$ lower bound even for the
\WL1-invariance ratio of $\gamma$ over $n$-vertex graphs. This implies
the integrality gap lower bound~\cite{ChappellGH17}, reproving it from
a different perspective. Moreover, our proof provides a solution to
Problem 3.4 in \cite{ChappellGH17} asking for an \emph{explicit} 
construction of graphs with logarithmic integrality gap for~$\gamma_f$.

Next, we consider the fractional edge-disjoint triangle packing number
$\rho^{K_3}_f$.  A general bound for the integrality gap of the
fractional matching number of a hypergraph \cite{Fueredi81} implies
that the integrality gap of $\rho^{K_3}_f$ is at most 2 (see Theorem
\ref{thm:ratio}).  This yields a polynomial-time algorithm
approximating $\rho^{K_3}$ within a factor of 2, which is competitive
with the greedy algorithm whose approximation ratio is 3; see
\cite{Yuster07}.\footnote{Though there are approaches
  \cite{HurkensS89} giving a better approximation ratio of
  $\frac32+\epsilon$, it is known \cite{FederS12} that there is no
  polynomial-time approximation scheme (PTAS) for $\rho^{K_3}$ unless
  $\mathrm{NP}=\mathrm{P}$.}  We observe that the upper bound of 2 is
sharp, as 2 is also a lower bound for the \WL2-invariance ratio
of~$\rho^{K_3}$.

Upper bounds for the \emph{additive} integrality gap of $\rho^{K_3}_f$
prove to be of considerable interest, implying a PTAS for $\rho^{K_3}$
on dense graphs \cite{HaxellR01,Yuster05}.  Motivated by this fact, in
Section \ref{s:inv} we obtain a lower bound also for the
\WL2-invariance \emph{difference} of~$\rho^{K_3}$.

\paragraph{\textbf{Related work.}}

Atserias and Dawar \cite{AtseriasD18} have shown that the
\WL1-invariance ratio for the vertex cover number $\tau$ is at most
2. Alternatively, this bound also follows from the \WL1-invariance of
$\nu_f$ (which implies the \WL1-invariance of $\tau_f$ as
$\tau_f=\nu_f$ by LP duality) combined with a standard rounding
argument. The argument presented in \cite{AtseriasD18} is 
different\footnote{The approach of \cite{AtseriasD18} is based on
  constructing weighted graphs $X_G$ and $X_H$ that are isomorphic if
  $G\eqkkwl1H$. The vertex cover number $\tau(G)$ is estimated from
  below and from above in terms of weighted vertex covers of $X_G$,
  and $\tau_f(G)$ appears as a technical tool in the argument.} and
alone does not yield \WL1-invariance of the fractional vertex
cover~$\tau_f$.

The bound of 2 for the \WL1-invariance ratio of $\tau$ is optimal.
Atserias and Dawar \cite{AtseriasD18} also show that the \WL
k-invariance ratio for $\tau$ is at least $7/6$ for each $k$. This
implies an unconditional inapproximability result for Vertex Cover in
the model of encoding-invariant computations expressible in FPC. It
remains open if similar lower bounds on the invariance ratios can be shown for
Dominating Set and Triangle Packing. 
For each parameter $\pi$ under consideration, we mainly
focus on \WL{k}-invariance for $k$ equal to the 
WL dimension of $\pi$. This focus is motivated by applications
to proving lower bounds for the integrality gap between $\pi$ and $\pi_f$
as discussed above.

\section{Reductions between linear programs}\label{s:LPs}

A \emph{linear program (LP)} is an optimization problem of the form
“maximize (or minimize) $a^tx$ subject to $Mx\le b$”, where
$a\in\Rset^n$, $b\in\Rset^m$, and $M\in\Rset^{m\times n}$ is an $m\times n$ matrix.
The variable $x$ varies over all vectors in $\Rset^n$
with nonnegative entries (which we denote by $x \ge 0$). Any vector
$x$ satisfying the constraints $Mx \le b$, $x \ge 0$ is called a
\emph{feasible solution} and the function $x\mapsto a^tx$ is called
the \emph{objective function}. We denote an LP with parameters $a,M,b$
by $LP(a,M,b,opt)$, where $opt=\min$ if the problem is minimization,
and $opt=\max$ if it is maximization. The optimum value of the
objective function over all feasible solutions is called the
\emph{value} of the program $L=LP(a,M,b,opt)$ and denoted by~$val(L)$.

Our goal now is to introduce an equivalence relation between LPs
ensuring equality of their values. We begin with a motivating discussion.

\paragraph{\textbf{Isomorphic and isometric LPs.}}

By duality, we can restrict our attention to maximization problems.
For the present discussion, we consider linear programs
$L_1=LP(a,M,b,\max)$ such that all feasible solutions $x$ satisfy
$Mx=b$.  This is ensured by the standard construction of augmenting
the LP with slack variables. With this LP we associate the linear
transformation $\alpha\function{\Rset^n}{\Rset^m}$ defined by
$\alpha(x)=Mx$. Let $L_2=LP(c,N,d,\max)$ be another LP with associated
linear transformation $\beta\function{\Rset^n}{\Rset^m}$, where
$\beta(x)=Nx$.

We say that $L_1$ and $L_2$ are \emph{isomorphic} if the following
two conditions hold:

\begin{enumerate}[(A)]
\item There is a permutation $\phi$ of the $n$ variables
  $x_1,\ldots,x_n$ and a permutation $\psi$ of the $m$ equations such
  that if $Z\in\Rset^{n\times n}$ and $Y\in\Rset^{m\times m}$ are the
  permutation matrices corresponding to $\phi$ and $\psi$
  respectively, then the linear transformations $\phi(x)=Zx$ and
  $\psi(w)=Yw$ make the following diagram
$$
\begin{CD}
\Rset^n @>{\alpha}>>\Rset^m@.\\
@A{\phi}AA @AA{\psi}A@.\\
\Rset^n @>>{\beta}>\Rset^m@.
\end{CD}
$$
commute. That is, $\alpha\phi=\psi\beta$.
\end{enumerate}
The second condition is
\begin{enumerate}[(A)]\setcounter{enumi}{1}
\item
$\phi(c)=a$ and $\psi(d)=b$.
\end{enumerate}

Suppose that $L_1$ and $L_2$ are isomorphic. If $x$ is a feasible
solution of $L_2$, which means $\beta(x)=d$, then $x'=\phi(x)$ is a
feasible solution of $L_1$. Indeed,
\begin{equation}
  \label{eq:alphax}
  \alpha(x')=\alpha(\phi(x))=\psi(\beta(x))=\psi(d)=b.
\end{equation}
Moreover, 
\begin{equation}\label{eq:atx}
  a^tx'=\phi(c)^t\phi(x)=c^tx,
\end{equation}
implying that $val(L_1)\ge val(L_2)$. Since the isomorphism of LPs is
clearly an equivalence relation, it follows by symmetry that
$val(L_1)=val(L_2)$.

More generally, suppose there are matrices $Y\in\mathbb{R}^{n\times m}$
and $Z\in\mathbb{R}^{m\times n}$ such that they satisfy the
orthogonality conditions $Y^tY=I_m$ and $Z^tZ=I_n$ (where $I_n$ is the
$n\times n$ identity matrix). Then we say $L_1$ and $L_2$
\emph{isometric} if conditions (A) and (B) are fulfilled with
$\phi(x)=Zx$ and $\psi(w)=Yw$. Equations \refeq{alphax} and
\refeq{atx} remain true as
$$
\phi(c)^t\phi(x)=(Zc)^tZx=c^tZ^tZx=c^tI_nx=c^tx.
$$

Hence, $val(L_1)=val(L_2)$ also holds for isometric $L_1$ and
$L_2$. This follows from the symmetry of the isometry relation, which
in its turn follows from the orthogonality conditions satisfied by
matrices $Y$ and $Z$.  

We now consider an equivalence concept for LPs, more general than the
orthogonality conditions, ensuring the equality of LP values.

\paragraph{\textbf{Equivalence of LPs.}}
  Let $L_1=LP(a,M,b,opt)$ and $L_2=LP(c,N,d,opt)$ be linear programs
  (in general form), where $a,c\in\Rset^n$, $b,d\in\Rset^m$,
  $M,N\in\Rset^{m\times n}$ and $opt\in\Set{\min,\max}$. We say that
  $L_1$ \emph{reduces} to $L_2$ ($L_1\le L_2$ for short), if
  there are matrices $Y\in\Rset^{m\times m}$ and $Z\in\Rset^{n\times
    n}$ such that
  \begin{itemize}
  \item $Y,Z\ge0$\vspace{-2mm}
  \item $a^tZ \;\diamondsuit\; c^t$, where $\diamondsuit =\;
    \begin{cases}
      \le,&\text{if } opt=\min \\
      \ge,&\text{if } opt=\max
    \end{cases}$\vspace{-2mm}
  \item $MZ\le YN$
  \item $Yd\le b$
  \end{itemize}
  $L_1$ and $L_2$ are said to be \emph{equivalent} if $L_1\le L_2$ and $L_2\le L_1$.

\begin{theorem}\label{thm:lpred}
  Equivalent linear programs $L_1$ and $L_2$ have equal values, i.e.,
  $val(L_1)=val(L_2)$.
\end{theorem}

\begin{proof}
  Let $L_1=LP(a,M,b,opt)$ and $L_2=LP(c,N,d,opt)$ and assume $L_1\le
  L_2$ via $(Y,Z)$. We show that for any feasible solution $x$ of
  $L_2$ we get a feasible solution $x'=Zx$ of $L_1$ with
  $a^tx'\;\diamondsuit\; c^tx$, where the relation symbol
  $\diamondsuit$ is as defined above. Indeed,
  \[
    Mx'=\underbrace{MZ\!}_{\le YN}x \le Y\!\underbrace{\!Nx}_{\le d}\le
    Yd \le b\text{ ~and~ }
    a^tx'=\underbrace{a^tZ}_{\diamondsuit\; c^t}x\;\diamondsuit\; c^tx.
  \]
  Thus, $L_1\le L_2$ implies $val(L_1)\;\diamondsuit\; val(L_2)$ and
  the theorem follows.
\end{proof}

Note that isometric LPs are equivalent. We now describe a different
kind of equivalent LPs.

\paragraph{\textbf{LPs with fractionally isomorphic matrices.}}
Recall that a square matrix $X\ge0$ is \emph{doubly stochastic} if its
entries in each row and column sum up to 1.  We call two $m\times n$
matrices $M$ and $N$ \emph{fractionally isomorphic} if there are
doubly stochastic matrices $Y\in\Rset^{m\times m}$ and
$Z\in\Rset^{n\times n}$ such that
\begin{equation}
  \label{eq:fim}
  MZ=YN\text{ and }NZ^t=Y^tM.
\end{equation}
Grohe et al.~\cite[Eq.~(5.1)-(5.2)]{GroheKMS14} discuss similar
definitions. Their purpose is to use fractional isomorphism and color
refinement to reduce the dimension of linear equations and LPs.
The meaning of \refeq{fim} will be clear from the proof of Theorem \ref{thm:set-pack} below.

\begin{lemma}\label{lem:fiMN}
  If $M$ and $N$ are fractionally isomorphic $m\times n$ matrices,
  then the linear programs $LP(\one_n,M,\one_m,opt)$ and
  $LP(\one_n,N,\one_m,opt)$ are equivalent, where $\one_n$ denotes the
  $n$-dimensional all-ones vector.
\end{lemma}

\begin{proof}
  Since the matrices $Y$ and $Z$ in \refeq{fim} are doubly stochastic,
  $Y\one_m=\one_m$ and $\one_n^tZ=\one_n^t$. Along with the first
  equality in \refeq{fim}, these equalities imply that $L_1\le
  L_2$. The reduction $L_2\le L_1$ follows similarly from the second
  equality in \refeq{fim} as $Y^t$ and $Z^t$ are doubly stochastic.
\end{proof}

\section{The Weisfeiler-Leman algorithm: Notation and formal definitions}

We consider undirected graphs, possibly with colored vertices.  The
vertex set and the edge set of a graph $G$ are denoted by $V(G)$ and
$E(G)$ respectively.  For $\barx=(x_1,\dots,x_k)$ in $V(G)^k$, let
$\alg k0G\barx$ be the $k\times k$ matrix $(m_{i,j})$ with $m_{i,j}=1$
if $x_ix_j\in E(G)$, $m_{i,j}=2$ if $x_i=x_j$ and $m_{i,j}=0$
otherwise.  We also augment $\alg k0G\barx$ by the vector of the
colors of $x_1,\dots,x_k$ if the graph $G$ is vertex-colored.  $\alg
k0G\barx$ encodes the \emph{ordered isomorphism type} of $\barx$ in
$G$ and serves as an initial coloring of $V(G)^k$ for \WL k.  In the
$r^{th}$ refinement round, \WL k computes a coloring $\alg krG\cdot$
of the Cartesian power $V(G)^k$ such that, if $\Pa_r$ is the color
partition of $V(G)^k$ according to $\alg krG\cdot$, then $\Pa_{r+1}$
is finer than or equal to $\Pa_r$ for every $r\ge0$.  Specifically,
\WL1 computes $\alg 1{r+1}Gx=(\alg 1{r}Gx,\Mset{\alg 1{r}Gy}{y\in
  N(x)})$, where $N(x)$ is the neighborhood of $x$ and
$\{\!\!\{\,\}\!\!\}$ denotes a multiset. If $k\geq 2$, \WL k refines
the coloring by $\alg k{r+1}Gx=(\alg krG\barx,\Mset{ (\alg
  krG{\barx_1^u},\dots,\alg krG{\barx_k^u}}{u\in V(G)})$, where
$\barx_i^u$ is the tuple $(x_1,\dots,x_{i-1},u,x_{i+1},\dots,x_k)$.
If $G$ has $n$ vertices, the color partition $\Pa_r$ stabilizes in at
most $n^k$ rounds. We define $\algkstab G\barx=\alg k{n^k}G\barx$ and
$\algkstabb G=\Mset{\algkstab G\barx}{\barx\in V(G)^k}$.  Now,
$G\eqkwl H$ if $\algkstabb G=\algkstabb H$.

The color partition of $V(G)$ according to $\algstabbb1(G,x)$ is
\emph{equitable}: for any color classes $C$ and $C'$, each vertex in
$C$ has the same number of neighbors in $C'$. Moreover, if $G$ is
vertex-colored, then the original colors of all vertices in each $C$
are the same. It is known \cite[Theorem~6.5.1]{ScheinermanU97} that $G\eqkkwl1 H$ 
exactly when $G$ and $H$ have a common equitable partition with the same neighborhood 
numbers in both graphs, after a suitable identification of vertex sets $V(G)$ and $V(H)$
(the coarsest such partition is actually the partition defined by the coloring~$\algstabbb 1 (G,\cdot)$).

Let $G$ and $H$ be graphs with vertex set $\{1,\ldots,n\}$, and let
$A$ and $B$ be the adjacency matrices of $G$ and $H$,
respectively. Then $G$ and $H$ are isomorphic if and only if $AX=XB$
for some $n\times n$ permutation matrix $X$. The linear programming
relaxation allows $X$ to be a doubly stochastic matrix.  If such an
$X$ exists,
$G$ and $H$ are said to be \emph{fractionally isomorphic}.
If $G$ and $H$ are colored graphs with the same partition of the vertex set into color classes, 
then it is additionally required that $X_{u,v}=0$ whenever $u$ and $v$ are of different colors.
Building on \cite{Tinhofer86}, it is shown by \cite{RamanaSU94} that two graphs
are indistinguishable by \WL1 if and only if they are
fractionally isomorphic (see also~\cite[Theorem 6.5.1]{ScheinermanU97}).

\section{Getting started}\label{s:starter}

\subsection{Fractional Set Packing}\label{ss:FSP}

The \emph{Set Packing} problem is to maximize the number of pairwise
disjoint sets in a given family of sets $\calS=\{S_1,\ldots,S_n\}$,
where each $S_j\subset\{1,\ldots,m\}$. The maximum is called in
combinatorics the \emph{matching number} of the hypergraph $\calS$ and
denoted by $\nu(\calS)$.  The fractional version of the matching
number can be expressed as a linear program
$LP(\calS)=LP(\mathbbm{1}_n,M,\mathbbm{1}_m,\max)$ where $M$ is the
$m\times n$ incidence matrix of $\calS$:
\begin{eqnarray*}
\max \sum_{i=1}^n x_i&&\text{under} \\
x_i&\ge&0\text{ for every }i\le n,\\
\sum_{i\,:\,S_i\ni j}x_i&\le&1\text{ for every }j\le m.
\end{eqnarray*}
The optimum value 
$$
\nu_f(\calS)=val(LP(\calS))
$$ 
is called the \emph{fractional matching number of}~$\calS$.

Let $\ig\calS$ denote the incidence graph of $\calS$. That is,
$\ig\calS$ is the vertex-colored bipartite graph with biadjacency
matrix $M$, where the bipartition of the vertex set is defined by $m$
red vertices and $n$ blue vertices. A red vertex $j$ is adjacent to a
blue vertex $i$ if $j\in S_i$.

\begin{theorem}\label{thm:set-pack}
  Let $\calS_1$ and $\calS_2$ be two families each consisting of $n$
  subsets of the set $\{1,\ldots,m\}$. If
  $\ig{\calS_1}\eqkkwl1\ig{\calS_2}$, then
  $\nu_f(\calS_1)=\nu_f(\calS_2)$.
\end{theorem}

\begin{proof}
  Denote the incidence matrices of $\calS_1$ and $\calS_2$ by $M$ and
  $N$ respectively.  Then
$$
A_1=
\begin{pmatrix}
0&M\\
M^t&0  
\end{pmatrix}
\text{ and }
A_2=
\begin{pmatrix}
0&N\\
N^t&0  
\end{pmatrix}
$$
are the adjacency matrices of $\ig{\calS_1}$ and $\ig{\calS_2}$
respectively.  Since $\ig{\calS_1}$ and $\ig{\calS_2}$ are
indistinguishable by color refinement,
these graphs are fractionally
isomorphic, that is, there is a doubly stochastic matrix $X$ such that
\begin{equation}
  \label{eq:frac-iso}
A_1X=XA_2  
\end{equation}
and $X_{uv}=0$ whenever $u$ and $v$ are from different vertex color
classes.  The latter condition means that $X$ is the direct sum of an
$m\times m$ doubly stochastic matrix $Y$ and an $n\times n$ doubly
stochastic matrix $Z$, that is,
$$
X=
\begin{pmatrix}
Y&0\\
0&Z  
\end{pmatrix}.
$$
Therefore, Equality \refeq{frac-iso} reads
$$
\begin{pmatrix}
0&M\\
M^t&0  
\end{pmatrix}
\begin{pmatrix}
Y&0\\
0&Z  
\end{pmatrix}
=
\begin{pmatrix}
Y&0\\
0&Z  
\end{pmatrix}
\begin{pmatrix}
0&N\\
N^t&0  
\end{pmatrix},
$$
yielding
$$
MZ=YN\text{ and }M^tY=ZN^t,
$$
that is, the matrices $M$ and $N$ are fractionally isomorphic.
Lemma \ref{lem:fiMN} implies that $LP(\calS_1)$ and $LP(\calS_2)$ are equivalent. Therefore, these LPs have equal values
by Theorem~\ref{thm:lpred}.
\end{proof}

The dual version of $LP(\calS)$ is the following minimization problem:
\begin{eqnarray*}
\min \sum_{j=1}^m y_j&&\text{under} \\
y_j&\ge&0\text{ for every }j\le m,\\
\sum_{j\in S_i}y_j&\ge&1\text{ for every }i\le n.
\end{eqnarray*}
This is an LP relaxation of the Hitting Set problem: Find a smallest set $Y\subset\{1,\ldots,m\}$
(called a \emph{hitting set}, \emph{cover}, or \emph{transversal})
having a non-empty intersection with each $S_i$.
Denote the optimum value by $\tau_f(\calS)$ and note that $\tau_f(\calS)=\nu_f(\calS)$
by LP duality.

\subsection{\WL1-invariance of the fractional domination number}\label{ss:gamma_f}

The \emph{closed neighborhood} of a vertex $x$ is defined as
$N[x]=N(x)\cup\{x\}$.  A set $D\subseteq V(G)$ is \emph{dominating} in
$G$ if $V(G)=\bigcup_{x\in D}N[x]$.  The \emph{domination number}
$\gamma(G)$ is the minimum cardinality of a dominating set in~$G$.

As a warm-up example, consider the fractional Dominating Set problem,
whose \WL1-invariance was established in~\cite{Rubalcaba-th}:
\begin{eqnarray*}
\min \sum_{v\in V(G)} y_v&&\text{under} \\
y_v&\ge&0\text{ for every }v\in V(G),\\
\sum_{v\in N[u]}y_v&\ge&1\text{ for every }u\in V(G).
\end{eqnarray*}

The value of this LP is the \emph{fractional domination number}
$\gamma_f(G)$. We can see this as the fractional Hitting Set problem
for $\calS=\calS_G$ consisting of the closed neighborhoods of all
vertices in $G$. The incidence matrix $M$ of $\calS_G$ and the
adjacency matrix $A$ of the graph $G$ are related by the equality
$M=A+I$. If $G\eqkkwl1 H$, then $G$ and $H$ are fractionally
isomorphic. That is, $AX=XB$ for a doubly stochastic $X$, where $B$ is
the adjacency matrix of $H$. It follows that $MX=(A+I)X=X(B+I)=XN$,
where $N$ is the incidence matrix of $\calS_H$. Similarly,
$X^tM=NX^t$. Therefore, $\gamma_f(G)=\gamma_f(H)$ by Lemma
\ref{lem:fiMN} and Theorem \ref{thm:lpred}. This follows also from
Theorem \ref{thm:set-pack} as $\ig{\calS_G}\eqkkwl1\ig{\calS_H}$ and
$\gamma_f(G)=\tau_f(\calS_G)=\nu_f(\calS_G)$ by LP duality.

\subsection{WL invariance through first-order interpretability}\label{ss:excess}

As we have just seen, given an instance graph $G$ of the fractional
Dominating Set problem, we can define an instance $\calS_G$ of the
fractional Hitting Set problem having the same LP value. The next
definition formalizes a general setting that is applicable to
essentially any logical formalism.

\begin{definition}
  We say that an instance $\calS_G$ of Fractional Set Packing or its
  dual version is definable over a graph $G$ with \emph{excess} $e$ if
$$
G\eqkkwl{(1+e)} H\implies \ig{\calS_G}\eqkkwl1\ig{\calS_H}.
$$
\end{definition}

This definition is very general. It includes settings where the
incidence graph $\ig{\calS_G}$ is \emph{first-order interpretable} in
the graph $G$ in the sense of \cite[Chapter 12.3]{EbbinghausF-b}.  In
other words, both the color predicate of $\ig{\calS_G}$ (defining the
red/blue bipartition of $\ig{\calS_G}$) and the adjacency relation of
$\ig{\calS_G}$ are first-order expressible in terms of the adjacency
relation of $G$, on $k$-tuples of vertices $V(G)^k$ for some $k$. The
number $k$ is the \emph{width} of the interpretation. In such a case,
if there is a first-order sentence over $s$ variables that is true on
$\ig{\calS_G}$ but false on $\ig{\calS_H}$, then there is a
first-order sentence over $sk$ variables that is true on $G$ but false
on $H$. Now, by the Cai-Fürer-Immerman result \cite{CaiFI92} that two
structures are $\kWL$-equivalent iff they are equivalent in the
$(k+1)$-variable counting logic $C^{k+1}$, we obtain the following
corollary from Theorem \ref{thm:set-pack}.

\begin{corollary}\label{cor:interpr}
  Let $\pi_f$ be a fractional graph parameter such that
  $\pi_f(G)=\nu_f(\calS_G)$, where $\calS_G$ admits a first-order
  interpretation (possibly with counting quantifiers) of width $k$ in
  $G$. Then $\calS_G$ is definable over $G$ with excess $2(k-1)$ and,
  consequently, $\pi_f$ is \WL{(2k-1)}-invariant.
\end{corollary}

\begin{remark}
  In order to obtain \WL1-invariance via Theorem~\ref{thm:set-pack},
  we need definability with zero excess. Applying Corollary
  \ref{cor:interpr} for this purpose would require a first-order
  interpretation of width 1, which may not always be
  possible. However, this is not the only way to get zero excess. 

  As an example (in a slightly general setting), consider
  $LP(\one_n,A^2,\one_n,opt)$ where $A$ is the adjacency matrix of
  $G$. As easily seen, if $G\eqkkwl1 H$, then there is a doubly
  stochastic $X$ such that $A^2X=AXB=XB^2$, and also $B^2X^t=X^tA^2$
  by taking transpose. Therefore, the value of
  $LP(\one_n,A^2,\one_n,opt)$ is \WL1-invariant.\footnote{Indeed, as
    observed by Rubalcaba~\cite{Rubalcaba-th}, this holds for any
    polynomial in $A$} The logical route to show \WL1-invariance via
  Corollary \ref{cor:interpr} is not feasible because the $(i,j)^{th}$
  entry of $A^2$ counts the number of 2-walks between vertices $i$ and
  $j$, which cannot be captured by the logic $C^2$. 

  Another example where a combinatorial argument yields more than
  Corollary \ref{cor:interpr} is presented below.
\end{remark}

\subsection{\WL1-invariance of the fractional matching number}

Recall that a set of edges $M\subseteq E(G)$ is a \emph{matching} in a
graph $G$ if every vertex of $G$ is incident to at most one edge from
$M$. The \emph{matching number} $\nu(G)$ is the maximum size of a
matching in~$G$.  Note that this terminology and notation agrees with
Section \ref{ss:FSP} when graphs are considered hypergraphs with
hyperedges of size 2.  The fractional Matching Problem is defined by
the LP
\begin{eqnarray*}
\max \sum_{uv\in E(G)} x_{uv}&&\text{under} \\
x_{uv}&\ge&0\text{ for every }uv\in E(G),\\
\sum_{v\in N(u)}x_{uv}&\le&1\text{ for every }u\in V(G),
\end{eqnarray*}
whose value is the \emph{fractional matching number} $\nu_f(G)$.  The
above LP is exactly the linear program $LP(\calS_G)$ for the instance
$\calS_G=E(G)$ of Fractional Set Packing formed by the edges of $G$ as
2-element subsets of $V(G)$, that is, $\nu_f(G)=\nu_f(\calS_G)$. A
first-order interpretation of the incidence graph $\calS_G$ in the
input graph $G$ of width 1 is clearly impossible. Moreover, an
interpretation of width 2 can only give \WL3-invariance by
Corollary \ref{cor:interpr}. Nevertheless, we can directly show that
$\calS_G$ is definable over $G$ with zero excess.

\begin{theorem}\label{thm:nu_f}
 The fractional matching number is \WL1-invariant. 
\end{theorem}

\begin{proof}
Given $G\eqkkwl1H$, we have to prove that $\nu_f(G)=\nu_f(H)$ or, equivalently,
$\nu_f(\calS_G)=\nu_f(\calS_H)$ where $\calS_G$ is as defined above.
By Theorem \ref{thm:set-pack}, it suffices to show that
$\ig{\calS_G}\eqkkwl1\ig{\calS_H}$. 
To this end, we construct a common equitable partition of $\ig{\calS_G}$ and $\ig{\calS_H}$,
appropriately identifying their vertex sets. Recall that
$V(\ig{\calS_G})=V(G)\cup E(G)$ and a red vertex $x\in V(G)$ is adjacent to a blue
vertex $e\in E(G)$ if $x\in e$.

For $x\in V(G)$, let $c_G(x)=\algstabbb1(G,x)$
and define $c_H$ on $V(H)$ similarly. First, we identify $V(G)$ and $V(H)$
(i.e., the red parts of the two incidence graphs) so that $c_G(x)=c_H(x)$
for every $x$ in $V(G)=V(H)$, which is possible because \WL1-equivalent graphs
have the same color palette after color refinement.
The color classes of $c_G$ now form a common equitable partition of $G$ and $H$.

Next, extend the coloring $c_G$ to $E(G)$ (the blue part of $\ig{\calS_G}$) 
by $c_G(\{x,y\})=\{c_G(x),c_G(y)\}$, and similarly extend $c_H$ to $E(H)$. 
Denote the color class of $c_G$ containing $\{x,y\}$ by $C_G(\{x,y\})$,
the color class containing $x$ by $C_G(x)$ etc. Note that $|C_G(\{x,y\})|$
is equal to the number of edges in $G$ between $C_G(x)$ and $C_G(y)$
(or the number of edges within $C_G(x)$ if $c_G(x)=c_G(y)$).
Since $\{C_G(x)\}_{x\in V(G)}$ is a common equitable partition of $G$ and $H$,
we have $|C_G(\{x,y\})|=|C_H(\{x',y'\})|$ whenever $c_G(\{x,y\})=c_H(\{x',y'\})$
(note that for any edge $\{x,y\}$ of $G$ there is an edge $\{x',y'\}$ 
of the same color in $H$ and vice versa).
This allows us to identify $E(G)$ and $E(H)$ so that $c_G(e)=c_H(e)$ for every $e$ in $E(G)=E(H)$.

Now, consider the partition of $V(G)\cup E(G)$ into the color classes of $c_G$ 
(or the same in terms of $H$) and verify that this is a common equitable
partition of $\ig{\calS_G}$ and $\ig{\calS_H}$. Indeed, let 
$C$ and $D$ be color classes of $c_G$ (hence, also of $c_H$).
It is enough to consider the case $C\subseteq V(G)$ and $D\subseteq E(G)$ 
because otherwise $C$ and $D$ are in the same part of the incidence graph
and there is no edge in between.
Consider an arbitrary vertex $x\in C$ and an arbitrary edge $e\in D$.
Let $e=\{u,v\}$. If neither $u$ nor $v$ is in $C$, then there is no edge
between $C$ and $D$ in both $\ig{\calS_G}$ and $\ig{\calS_H}$.
Otherwise, suppose that $u\in C$ and denote $C'=C_G(v)$ (it is not excluded that $C'=C$). 
Clearly, the vertex $x$ has exactly as many $D$-neighbors in $\ig{\calS_G}$
as it has $C'$-neighbors in $G$. This number depends only on $C$ and $C'$
or, equivalently, only on $C$ and $D$, and is the same if counted for~$\ig{\calS_H}$.

On the other hand, $e$ has exactly one $C$-neighbor, $u$, in $\ig{\calS_G}$ and $\ig{\calS_H}$ 
if $C'\ne C$ and exactly two $C$-neighbors, $u$ and $v$, if $C'=C$. 
What is the case depends only on $D$ and $C$, and is the same in $\ig{\calS_G}$ and $\ig{\calS_H}$.
Thus, we do have a common equitable partition of $\ig{\calS_G}$ and~$\ig{\calS_H}$.
\end{proof}

The fractional matching number is precisely the fractional
$K_2$-packing number, and we generalize Theorem \ref{thm:nu_f} to
fractional $F$-packing numbers in Section~\ref{sec:pack}.  In
particular, there we will establish \WL2-invariance of Fractional
Triangle Packing. The approach we used in
the proof of Theorem \ref{thm:nu_f} works as well for \emph{edge-disjoint}
packing, which we demonstrate in the next subsection.

\subsection{\WL2-invariance of Fractional Edge-Disjoint Triangle Packing}\label{ss:EDT}

Given a graph $G$, let $T(G)$ denote the family of all sets $\{e_1,e_2,e_3\}$
consisting of the edges of a triangle subgraph in $G$.
We regard $T(G)$ as a family $\calS_G$ of subsets of the edge set $E(G)$.
The optimum value of Set Packing Problem on $\calS_G$, which we denote by $\rho^{K_3}(G)$,
is equal to the maximum number of edge-disjoint triangles in $G$.
Let $\rho^{K_3}_f(G)=\nu_f(\calS_G)$ be the corresponding fractional parameter.

\begin{theorem}\label{thm:rho_f}
 The fractional edge-disjoint triangle packing number $\rho^{K_3}_f$ is \WL2-invariant. 
\end{theorem}

\begin{proof}
  Given a graph $G$, we consider the coloring $c_G$ of $E(G)\cup T(G)$
  defined by $c_G(\{x,y\})=\{\algstabbb2(G,x,y),\algstabbb2(G,y,x)\}$
  on $E(G)$ and
  $c_G(\{e_1,e_2,e_3\})=\{\!\!\{ c_G(e_1),\allowbreak c_G(e_2),c_G(e_3) \}\!\!\}$
  on $T(G)$.  As in the proof of Theorem \ref{thm:nu_f}, the upper
  case notation $C_G(s)$ will be used to denote the color class
  containing $s\in E(G)\cup T(G)$.

Suppose that $G\eqkkwl2H$. This condition implies that we can identify
the sets $E(G)$ and $E(H)$ so that $c_G(e)=c_H(e)$ for every $e$ in
$E(G)=E(H)$.  Moreover, the \WL2-equivalence of $G$ and $H$ implies
that for any $t$ in $T(G)$ there is $t'$ of the same color in $T(H)$ and vice versa.
What is more, for any $t\in T(G)$ and $t'\in T(H)$ with
$c_G(t)=c_H(t')$ we have $|C_G(t)|=|C_H(t')|$. This allows us to identify $T(G)$ and $T(H)$ so
that $c_G(t)=c_H(t)$ for every $t$ in $T(G)=T(H)$. As in the proof of
Theorem \ref{thm:nu_f}, it suffices to argue that $\{C_G(w)\}_{w\in
  E(G)\cup T(G)}$ is a common equitable partition of the incidence
graphs $\ig{\calS_G}$ and $\ig{\calS_H}$. The equality
$\rho^{K_3}_f(G)=\rho^{K_3}_f(H)$ will then follow by
Theorem~\ref{thm:set-pack}.

Let $C\subseteq E(G)$ and $D\subseteq T(G)$ be color classes of $c_G$ (hence, also of $c_H$).
Consider an arbitrary triangle $t=\{e_1,e_2,e_3\}$ in $D$.
If none of $e_i$ belongs to $C$, then there is no edge
between $C$ and $D$ in both $\ig{\calS_G}$ and $\ig{\calS_H}$.
Suppose that $e_1\in C$ and denote $C'=C_G(e_2)$ and
$C''=C_G(e_3)$ (it is not excluded that some of the classes $C$, $C'$,
and $C''$ coincide).

Denote the vertices of $t$ by $u,v,w$ and suppose that $e_1=\{u,v\}$.
Consider an arbitrary edge $e=\{x,y\}$ in $C$. Let us count the number of $D$-neighbors that $e$ has in
$\ig{\calS_G}$, that is, the number of triangles in $G$ with one edge $e$ and 
two other edges $e'\in C'$ and $e''\in C''$.
This number is equal to the number of vertices $z$ such that
$(\algstabbb2(G,x,z),\algstabbb2(G,z,y))$ is one of the 
$8$ pairs in $(c_G(\Set{u,w})\times c_G(\Set{v,w}))\cup
(c_G(\Set{v,w})\times c_G(\Set{u,w}))$, like
$(\algstabbb2(G,w,v),\algstabbb2(G,u,w))$ (some of these pairs can coincide).
Since the partition of $V(G)^2$ by the coloring
$\algstabbb2(G,\cdot,\cdot)$ is not further refined by \WL2, this
number does not depend on the choice of $e$ in $C$, depending only on $C$ and $D$.  We obtain the same number also
while counting the $D$-neighbors of $e$ in~$\ig{\calS_H}$.

On the other hand, $t$ has exactly one neighbor $e_1$ in $C$ if $C$
differs from both $C'$ and $C''$, exactly two $C$-neighbors if $C$
coincides with exactly one of $C'$ and $C''$, and exactly three
$C$-neighbors $e_1$, $e_2$, and $e_3$ if $C=C'=C''$.  Which of the three
possibilities occurs depends only on $D$ and $C$, and is the same in
$\ig{\calS_G}$ and $\ig{\calS_H}$.  This completes our analysis, showing
that we really have a common equitable partition of $\ig{\calS_G}$ and $\ig{\calS_H}$.  
\end{proof}

Note that \WL2-invariance in Theorem \ref{thm:rho_f} is optimal. Indeed, 
for the \WL1-equivalent graphs $2C_3$ and $C_6$ we have $\rho^{K_3}_f(2C_3)=2$ while
$\rho^{K_3}_f(C_6)=0$.

\section{The $F$-packing number}\label{sec:pack}

For graphs $F$ and $G$, let $\Sub FG$ denote the set of all
  subgraphs $S$ of $G$ that are isomorphic to $F$. An 
  \emph{$F$-packing} of $G$ is a set
  $P\subseteq\Sub FG$ where all subgraphs are vertex disjoint.
  The \emph{$F$-packing number} $\pi^F(G)$ is the
  maximum size of an $F$-packing $P$ of $G$. 
Let $\calS_{F,G}=\Set{V(S)}{S\in\Sub FG}$. Note that $\pi^F(G)=\nu(\calS_{F,G})$,
the matching number of the hypergraph $\calS_{F,G}$.
The \emph{fractional $F$-packing number} of $G$ is defined by $\pi_f^F=\nu_f(\calS_{F,G})$,
where $\nu_f$ is the fractional matching number of a hypergraph as introduced in Section~\ref{ss:FSP}.

The following parameter plays a key role in our approach to estimating the WL dimension of~$\pi^F_f$.
  We define the \emph{homomorphism-hereditary treewidth} of a graph $F$,
  denoted by $\htw F$, as the maximum treewidth $\tw{F'}$ over all
  homomorphic images $F'$ of~$F$, i.e., over all $F'$ such that there
  is a homomorphism $h$ from $F$ to $F'$ that is vertex and edge
  surjective. 

For a colored graph $S$, let $S^\circ$ denote the underlying uncolored version of~$S$.

We begin with an important technical lemma.  In what follows, $F$ is
an uncolored graph and $G$ is a graph endowed with a coloring $c\function{V(G)}{\Set{1,\dots,r}}$.  
For a subgraph $S$ of $G$, we define its
\emph{color type} as the multiset $\mu(S)=\Mset{c(u)}{u\in V(S)}$.
For a given color type $\mu$, we write $\eSubb \mu FG$ to denote the
set of all subgraphs $S$ of $G$ such that $S^\circ$ is isomorphic to $F$ and
$\mu(S)=\mu$.  The number of subgraphs $S\in\eSubb \mu FG$ containing
a vertex $x$ is denoted by $\esub x\mu FG$.

Given $x\in V(G)$, we set $\algkstab Gx=\algkstab G\barx$ where
$\barx$ is the $k$-tuple whose all elements are equal to~$x$.  
Dvořák \cite{Dvorak10} proves that, for each graph $F$ of treewidth $k$, 
the number of homomorphisms from $F$ to a graph $G$ is \kWL-invariant.
In particular \cite[Lemma 4]{Dvorak10}, if $\algkstab Gx=\algkstab Hy$
and $F$ is a graph of treewidth at most $k$ with a designated vertex $z$, then
the number of homomorphisms from $F$ to $G$ taking $z$ to $x$
is equal to the number of homomorphisms from $F$ to $H$ taking $z$ to $y$.
The classical result by Lov{\'{a}}sz \cite[Section 5.2.3] {hombook}
shows the number of subgraphs of a graph $G$ isomorphic to $F$
is determined by the numbers of homomorphisms from $F'$ to $G$,
where $F'$ ranges over homomorphic images of $F$.
Combined with Dvořák's result, this implies the following fact
(see \cite{shades} for some details).

\begin{lemma}\label{lem:uniform}
  If $\htw F\leq k$, then the count $\esub x\mu FG$ is determined by
  the quadruple $F$, $\mu$, $\algkstab Gx$, and $\algkstabb G$.
\end{lemma}

\noindent
Note that, if $k\ge2$, then $\algkstabb G\ne\algkstabb H$ even implies
that $\algkstabb G\cap\algkstabb H=\emptyset$. This means that the
stable color $\algkstab Gx$ of any single vertex $x$ determines the
color palette $\algkstabb G$. Thus, for $k\ge2$, Lemma
\ref{lem:uniform} says that $\esub x\mu FG$ is determined by the
triple $F$, $\mu$, and $\algkstab Gx$.

Denote the incidence graph of the hypergraph $\calS_{F,G^\circ}$ by $G^F$.
Recall that $G^F$ is a colored graph with two color classes $V(G)$ (red) and $\Sub F{G^\circ}$ (blue)
where a red vertex $x$ is adjacent to a blue vertex $S$ if $x\in V(S)$.

\begin{lemma}\label{lem:FextA}
  Let $G$ be a graph with vertex coloring $c$ and $\Pa=\Set{C_1,\dots,C_{r}}$ 
be the corresponding color partition of $V(G)$.
For a graph $F$, suppose that the counts
  $\esub x\mu FG$ only depend on $(c(x),\mu,F)$. Then the partition $\Pa$
  can be extended to an equitable partition
  $\Pa^F=\Set{C_1,\dots,C_{r+s}}$ of $G^F$.
Specifically, if $\mu_1,\dots,\mu_s$ are all color types $\mu$ with $\eSubb \mu FG\ne\emptyset$,
then $C_{r+i}=\Set{S^\circ}{S\in\eSubb {\mu_i}FG}$ for $i=1,\dots,s$.
\end{lemma}

\begin{proof}
For each pair $x,y$ of vertices in each $C_i$, we need to show that
$x$ and $y$ have the same number of neighbors in each $C_j$. If
$i,j\le r$ or $i,j > r$, then between $C_i$ and $C_j$
there is no edge at all.  If $i\leq r$ and $j>r$, then
$\esub x{\mu_{j-r}}FG=\esub y{\mu_{j-r}}FG$ as $c(x)=c(y)=i$.
Since $\esub x{\mu_{j-r}} FG$ is exactly the number of neighbors of
$x$ in $C_j$, the claim follows also in this case.  Finally, for $i>r$
and $j\leq r$ the color type $\mu_{i-r}$ contains the color
corresponding to $C_j$ with a certain multiplicity $m_j$ which for any
vertex $x\in C_i$ coincides with the number of its neighbors in $C_j$.
\end{proof}

\begin{theorem}\label{thm:htwpack}
  If $\htw F\leq k$, then $\pi_f^F$ is \WL k-invariant.
\end{theorem}

\begin{proof}
  Let $G$ and $H$ be two
  $\kWL$-equivalent graphs. Color the vertices of $G$ and $H$
  with their stable $\kWL$ vertex colors, that is, assign each vertex $x\in V(G)$
color $c(x)=\algkstab Gx$ and each vertex $x\in V(H)$
color $c(x)=\algkstab Hx$. Since $\algkstabb G=\algkstabb H$, we can identify the vertex sets
$V(G)$ and $V(H)$ so that each vertex $x$ has the same color $c(x)$ in both graphs.
Let $\Pa=\Set{C_1,\dots,C_{r}}$ be the corresponding color partition of $V(G)=V(H)$. 

Moreover, it follows
  from Lemma~\ref{lem:uniform} that the counts $\esub x\mu FG$ and
  $\esub x\mu{\allowbreak F}H$ coincide and are uniform within each vertex class $C_i$. 
This implies, in particular, that the color types $\mu_1,\dots,\mu_s$ of subgraphs $S$
with $S^\circ\cong F$ appearing in $G$ and in $H$ are the same.
Let $\Pa^F_G=\Set{C_1,\dots,C_{r+s}}$ and $\Pa^F_H=\Set{C'_1,\dots,C'_{r+s}}$ 
be the equitable partitions of $G^F$ and $H^F$ given by Lemma~\ref{lem:FextA}.
Recall that $x$ has $\esub x{\mu_i} FG=\esub x{\mu_i} FH$ neighbors in $C_{r+i}$ in the graph $G^F$
and the same number of neighbors in $C'_{r+i}$ in the graph $H^F$.
On the other hand, any two subgraphs $S\in C_{r+i}$ and $S'\in C'_{r+i}$ have the same
color type $\mu_i$ and, therefore, equally many neighbors in each $C_j$, $j\le r$,
in $G^F$ and in $H^F$. We conclude from here that $|C_{r+i}|=|C'_{r+i}|$ for every $i\le s$.
Therefore, $\Sub FG$ and $\Sub FH$ can be identified so that $C_{r+i}=C'_{r+i}$ for every $i\le s$,
and $\Pa^F=\Pa^F_G=\Pa^F_H$ becomes a common equitable partition of $G^F$ and~$H^F$. 

The existence of a common equitable partition implies that $G^F$ and $H^F$
are $\WL1$-equivalent. Applying Theorem \ref{thm:set-pack},
we conclude that $\pi_f^F(G)=\pi_f^F(H)$.
\end{proof}

\begin{remark}\label{rem:edge-disj}
  If subgraphs in $P\subseteq\Sub FG$ are allowed to share vertices
  but required to be edge disjoint, we call $P$ an \emph{edge disjoint
    $F$-packing} of $G$.  The \emph{edge disjoint $F$-packing number}
  $\rho^F$ is defined as the maximum size of such $P$. Redefine
  $\calS_{F,G}$ by $\calS_{F,G}=\Set{E(S)}{S\in\Sub FG}$, which now
  becomes a hypergraph on the set $E(G)$. Now,
  $\rho^F(G)=\nu(\calS_{F,G})$, and we also define the
  \emph{fractional edge-disjoint $F$-packing number} of $G$ by
  $\rho_f^F=\nu_f(\calS_{F,G})$. Similarly to Theorem
  \ref{thm:htwpack}, if $\htw F\leq k$, then $\rho_f^F$ is \WL
  k-invariant.  Basically the same argument works out, and we here
  only comment on some proof details (cf.\ also the proof of
  Theorem~\ref{thm:rho_f}).

  We set $\algkstab G{x,y}=\algkstab G{\bar z}$ where the pair $(x,y)$
  is extended to the $k$-tuple $\bar z=(x,y,\ldots,y)$ if $k>2$.  If
  $k=1$, we define $\alg 1{}G{x,y}=(\alg 1{}Gx,$ $\alg 1{}Gy)$.  For
  an edge $e=\{x,y\}$ of $G$, denote
$$
c(e)=\{\algkstab G{x,y},\algkstab G{y,x}\}.
$$
For a subgraph $S$ of $G$, we adapt the notion of the color type of $S$, now defining it
by $\mu(S)=\Mset{c(e)}{e\in E(S)}$.

Let $\esub e\mu FG$ denote the number of 
subgraphs $S\in\Sub FG$ having color type $\mu$ and containing an edge $e\in E(G)$.
An analog of Lemma \ref{lem:uniform} says that, for a pair of adjacent vertices $x$ and $y$,
  the count $\esub {\{x,y\}}\mu FG$ is determined by 
  the tuple $(\algkstab G{x,y},\mu,F)$ (note that $\algkstab G{x,y}$
determines $\algkstab G{y,x}$ and vice versa).
Lemma \ref{lem:FextA} has an analog for the edge coloring $c$
introduced above.
\end{remark}

\section{Invariance ratio and integrality gap}\label{s:inv}

\subsection{Edge-disjoint triangle packing: Invariance ratio}

We now turn back to the edge-disjoint triangle packing number $\rho^{K_3}$
and its fractional version $\rho^{K_3}_f$, defined in Section \ref{ss:EDT}.
Let
$$
\mathit{IG}^{K_3}=\sup_{G}\frac{\rho^{K_3}_f(G)}{\rho^{K_3}(G)},
$$
where the supremum is taken over all graphs containing at least one triangle,
be the integrality gap of $\rho^{K_3}_f$.
We define the \emph{invariance ratio} of $\rho^{K_3}$ by
$$
\mathit{IR}^{K_3}=\sup_{G,\,H}\left|\frac{\rho^{K_3}(G)}{\rho^{K_3}(H)}\right|
$$
where the supremum is taken over all pairs of \WL2-equivalent graphs $G$
and $H$ containing at least one triangle. 
If $G$ and $H$ are such graphs, then, by Theorem~\ref{thm:rho_f},
$$
\rho^{K_3}(G)\le\rho^{K_3}_f(G)=\rho^{K_3}_f(H)\le\mathit{IG}^{K_3}\cdot\rho^{K_3}(H),
$$
which implies the relationship
$$
\mathit{IG}^{K_3}\ge\mathit{IR}^{K_3}.
$$

\begin{theorem}\label{thm:ratio}
$\mathit{IG}^{K_3}=\mathit{IR}^{K_3}=2$.
\end{theorem}

\begin{proof}
We begin with proving a lower bound $\mathit{IG}^{K_3}\le2$.
In Section \ref{ss:FSP}, we discussed the fractional matching number $\nu_f(\calS)$ 
of a hypergraph~$\calS$. Füredi \cite{Fueredi81} proved tight bounds for
the integrality gap of this parameter. In particular, if a 3-uniform hypergraph $\calS$
does not contain any set of seven hyperedges forming the Fano plane, then
$\nu_f(\calS)/\nu(\calS)\le2$. Recall that a hypergraph is \emph{$r$-uniform}
if each hyperedge consists of $r$ vertices. The \emph{Fano plane} is the 3-uniform
hypergraph with 7 vertices and 7 hyperedges shown in Figure
\ref{fig:fano}(a).
Let $G$ be an arbitrary graph, and $\calS_G$ be the
hypergraph we associated with $G$ in Section \ref{ss:EDT}. 
Recall that $\rho^{K_3}_f(G)=\nu_f(\calS_G)$.
In order to prove the bound $\mathit{IG}^{K_3}\le2$,
it suffices to check that $\calS_G$ does not contain any copy of the Fano plane.

Indeed, assume that $F$ is a copy of the Fano plane in $\calS_G$.
Enumerate the vertices of $F$ as in Figure \ref{fig:fano}(a).
Each vertex $i$ of $F$ is an edge of $G$, which we denote by $e_i$.
Each hyperedge of $F$ consists of the edges of a triangle in $G$. 
Thus, the triangles of $G$ corresponding to the hyperedges
$\{1,2,3\}$ and $\{1,4,5\}$ share an edge $e_1$; see Figure \ref{fig:fano}(b).
The edge $e_6$ must, therefore, connect the two vertices of these
triangles not incident to $e_1$. Consequently, $e_1$ and $e_6$ are not adjacent,
contradicting the fact that they belong to the hyperedge $\{1,6,7\}$.

We complete the proof by showing that
$
\mathit{IR}^{K_3}\ge2.
$
Denote the Shrikhande and the $4\times4$ rook's graphs by
$S$ and $R$ respectively. Both have vertex set
  $\Zset_4\times \Zset_4$, and $(i,j)$ and $(i',j')$ are adjacent in $S$ if ($i=i'$
  and $j'=j+1$) or ($j=j'$ and $i'=i+1$) or
  ($i'=i+1$ and $j'=j+1$), where equality is in $\Zset_4$, while they are adjacent in $R$ if
  $i=i'$ (row 4-clique) or $j=j'$ (column 4-clique).  The Shrikhande graph
is completely decomposable into edge-disjoint triangles
  $\Set{(i,j),(i+1,j),(i+1,j+1)}$ and, hence, $\rho^{K_3}(S)=16$.  
On the other hand, in $R$ the edges of
  each $K_3$ all belong to the same row or column 4-clique, and
  the rest of the edges in this row/column correspond to a star. 
Since a packing can take at most one $K_3$ from each row/column,
we have $\rho^{K_3}(R)=8$. 
It remains to note that both $S$ and $R$ are strongly regular graphs with the same
parameters $(16,6,2,2)$. Therefore, $S\eqkkwl2R$.
\end{proof}

\begin{figure}[t]
\centering
\begin{tikzpicture}[
  lab/.style={draw=none,fill=none},
  every node/.style={circle,draw,black,
  inner sep=2pt,fill=black},
  lab/.style={draw=none,fill=none,inner sep=0pt,rectangle},
  longer/.style={shorten >=-4mm,shorten <=-4mm},
  vt/.style={line width=1.4pt},
  line width=0.2pt,
  outer sep=0pt,
]
  \node[isosceles triangle,dgreen,
    inner sep=6mm,
    isosceles triangle apex angle=60,
    shape border uses incircle,
    shape border rotate=90,vt,
    fill=none] (fanotr) {}
  ;
  \node[at=(fanotr.center),inner sep=6mm,fill=none,vt,dgreen] {};
  \path 
    (fanotr.center) node[label={right:$7$}] (0) {}
    (fanotr.30)     node[label={right:$2$}] (4) {}
    (fanotr.90)     node[label={above:$5$}] (1) {}
    (fanotr.150)    node[label={ left:$1$}] (2) {}
    (fanotr.210)    node[label={ left:$4$}] (6) {}
    (fanotr.270)    node[label={below:$3$}] (5) {}
    (fanotr.330)    node[label={right:$6$}] (3) {}
  ;
  \draw[vt,red] (2) -- (0) -- (3);
  \draw[vt]
    (4) -- (6) -- (0)
    (5) -- (0) -- (1)
  ;
  \path ($(1)+(3.5cm,0cm)$) node (e456) {}
  ++(1.7cm,-3mm) node (e125) {}
  ++(0cm,-1.7cm) node (e236) {}
  ++(-1.7cm,3mm) node (e134) {}
  ;
  \draw[vt] 
  (e125) to node[right=1mm,lab] {$e_2$} (e236)
  (e125) to node[right=2mm,lab] {$e_1$} (e134)
  (e134) to node[below=1mm,lab] {$e_3$} (e236)
  (e134) to node[right=1mm,lab]  {$e_4$} (e456)
  (e125) to node[above=1mm,lab] {$e_5$} (e456)
  (e456) .. %
  controls ($(e134)-(17mm,0mm)$) and  ($(e134)-(0mm,17mm)$) ..
  node[below=2mm,lab] {$e_6$}
  (e236)
  ;
  \path ($(6)-(7mm,4mm)$) node[lab] {(a)}
  ++(4.8cm,0cm) node[lab] {(b)};
\end{tikzpicture}
\caption{(a) Fano plane (b) The only constellation of triangles consistent with
the green hyperedges of the Fano plane. Adding a triangle for the
red hyperedge is impossible.}
\label{fig:fano}
\end{figure}
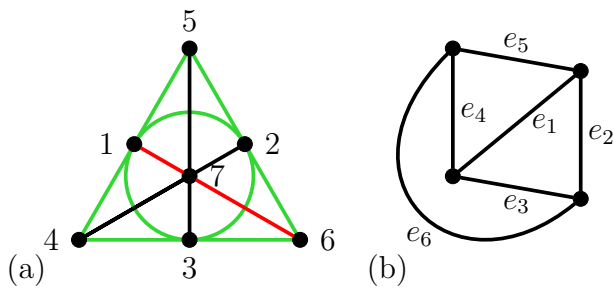

\subsection{Edge-disjoint triangle packing: Invariance difference}\label{ss:inv-diff}

Denote the number of vertices in a graph $G$ by $v(G)$. Let
$$
\mathit{AIG}^{K_3}(n)=\max_{G\,:\;v(G)=n}\of{\rho^{K_3}_f(G)-\rho^{K_3}(G)}
$$
be the \emph{additive integrality gap} of $\rho^{K_3}_f$.
Haxell and R{\"{o}}dl \cite{HaxellR01} proved that 
$$
\mathit{AIG}^{K_3}(n)=o(n^2),
$$
which gives a PTAS for $\rho^{K_3}$ on dense enough graphs.
On the other hand, Yuster \cite{Yuster05} showed that
$\mathit{AIG}^{K_3}(n)=\Omega(n^{1.5})$, and it is open whether this  lower bound is tight.
Define the \emph{invariance difference} of $\rho^{K_3}$ as the function 
$$
\mathit{ID}^{K_3}(n)=\max|\rho^{K_3}(G)-\rho^{K_3}(H)|
$$
where the maximum is taken over \WL2-equivalent $n$-vertex graphs $G$
and $H$. As follows from Theorem \ref{thm:rho_f}, $\mathit{ID}^{K_3}(n)$ provides
a lower bound for the additive integrality gap of $\rho^{K_3}_f$, namely
\begin{equation}
  \label{eq:AIGID}
\mathit{AIG}^{K_3}(n)\ge\frac12\,\mathit{ID}^{K_3}(n).
\end{equation}  
In this respect, it would be interesting to determine the asymptotics
of $\mathit{ID}^{K_3}(n)$ and to investigate how tight the relation \refeq{AIGID} is.
The following result is a step towards this goal.

\begin{theorem}\label{thm:discrep}
 $\mathit{ID}^{K_3}(n)=\Omega(n^{1.25})$. 
\end{theorem}

The proof uses two lemmas, which we state and prove now.
The \emph{tensor product} $G\times G'$ of graphs $G$ and $G'$ is the graph
on the vertex set $V(G)\times V(G')$ with vertices $(u,u')$ and $(v,v')$
adjacent if $u$ and $v$ are adjacent in $G$ and $u'$ and $v'$ are adjacent in~$G'$.

\begin{lemma}\label{lem:tensor-prod}
If $G\eqkkwl2 H$ and $G'\eqkkwl2 H'$, then $G\times G'\eqkkwl2 H\times H'$.
\end{lemma}

\begin{proof}
We use the fact \cite{Hella96} that $G \eqkkwl 2 H$ if and only if
Duplicator has a winning strategy in the 3-pebble \emph{Hella's bijection game} on $G$ and $H$. 
The game is played by two players, Spoiler and Duplicator. 
There are three pairwise distinct pebbles $p_1,p_2,p_3$,
each given in duplicate. In one round of the game, Spoiler
puts one of the pebbles $p_i$ on a vertex in $G$ and its copy on a
vertex in $H$. When $p_i$ is on the board, $x_i$ denotes the vertex
pebbled by $p_i$ in $G$, and $y_i$ denotes the vertex pebbled by the
copy of $p_i$ in $H$. Specifically, a round is played as follows:
\begin{itemize}
\item 
Spoiler chooses $i\in\{1,2,3\}$;
\item 
Duplicator responds with a bijection $f\function{V(G)}{V(H)}$ obeying the condition that
$f(x_j)=y_j$ for all $j\ne i$ such that $p_j$ is on the board;
\item Spoiler chooses a vertex $x$ in $G$ and puts $p_i$ on $x$ and
  its copy on $f(x)$ (this move reassigns $x_i$ to vertex $x$ and
  $y_i$ to vertex $f(x)$).
\end{itemize}

Duplicator wins if she manages to keep the map $x_i\mapsto y_i$ a partial
isomorphism during the play; otherwise the winner is Spoiler. 

The \WL2-equivalence of $G$ and $H$ and of $G'$ and $H'$ implies that
Duplicator has a winning strategy in the bijection game on $G$ and $H$ and on $G'$ and $H'$.
She can combine these strategies to win also the game on $G\times G'$ and $H\times H'$
by regarding it as simultaneous play of a game on $G$ and $H$ and a game on $G'$ and $H'$.
Whenever $\bar x_i=(x_i,x'_i)$ and $\bar y_i=(y_i,y'_i)$ are pebbled in $G\times G'$ and $H\times H'$,
Duplicator assumes that $x_i$ and $y_i$ are pebbled in the game on $G$ and $H$ while
$x'_i$ and $y'_i$ are pebbled in the game on $G'$ and $H'$. Given two bijections
$f$ and $f'$ for these games, she provides Spoiler with the bijection $\bar f$
defined by $\bar f(u,u')=(f(u),f'(u'))$. Since $f$ and $f'$ ensure partial isomorphisms
between $G$ and $H$ and between $G'$ and $H'$, the bijection $\bar f$ maintains
a partial isomorphism between $G\times G'$ and $H\times H'$.
The existence of a winning strategy implies that the product graphs are \WL2-equivalent.
\end{proof}

We say that a graph $G$ is \emph{$K_3$-decomposable} if there is an edge-disjoint triangle
packing covering all edges of $G$, that is, $\rho^{K_3}(G)=e(G)/3$.

\begin{lemma}\label{lem:decomp}
If both $G$ and $H$ are $K_3$-decomposable, then $G\times H$ is $K_3$-decomposable too.
\end{lemma}

\begin{proof}
Note first that the claim is true for $G=H=K_3$ because every edge in $K_3\times K_3$
has a unique extension to a triangle.

Let $T_G$ be a complete edge-disjoint triangle packing in $G$, and
$T_H$ be a complete edge-disjoint triangle packing in $H$. The set of all possible products $t\times t'$
over all triangles $t\in T_G$ and $t'\in T_H$ is a complete edge-disjoint $K_3\times K_3$-packing
in $G\times H$. The lemma follows as each $K_3\times K_3$ is $K_3$-decomposable.  
\end{proof}

\begin{proof}[Proof of Theorem \ref{thm:discrep}]
  Consider the Shrikhande and $4\times4$ rook's graphs, $S$ and $R$,
  as in the proof of Theorem \ref{thm:ratio}.  Since $S\eqkkwl2 R$, by
  Lemma \ref{lem:tensor-prod} we also have $S^k\eqkkwl2 R^k$, where
  the $k^{th}$ power is with respect to the tensor product.  To obtain
  the bound for $\mathit{ID}^{K_3}(n)$, it suffices, therefore, to
  prove that
\begin{equation}
  \label{eq:SkRk}
\rho(S^k)-\rho(R^k)=\Omega(v(R^k)^{1.25})=\Omega(2^{5k}).
\end{equation}

Recall that $S$ is $K_3$-decomposable. By Lemma \ref{lem:decomp}, $S^k$ is also decomposable
and, hence, $\rho(S^k)\ge\rho(R^k)$. Let $\partial(G)$ denote the number of edges that remain not covered
by an optimal triangle packing in $G$. Note that $\rho(S^k)-\rho(R^k)=\partial(R^k)$,
and we will estimate the last value.

Obviously, 
$$
\partial(G) \ge v_{\mathrm{odd}}(G)/2,
$$ 
where $v_{\mathrm{odd}}(G)$ denotes the number of vertices of odd degree in $G$.
Denote $K=(K_4)^k$. Since $v_{\mathrm{odd}}(K)=v(K)=4^k$, we conclude that $\partial(K)\ge 2^{2k-1}$.

Note that $R=K_4\square K_4$, where $\square$ denotes the Cartesian product of graphs.
It readily follows that $R$ is completely decomposable into 8 copies of $K_4$ and, therefore,
$R^k$ is completely decomposable into $8^k$ copies of $K$.
Every triangle $t$ in $R^k$ is included in one of these $K$-subgraphs.
Indeed, let $t_1,\ldots,t_k$ be the projections of $t$ onto the $k$ coordinates.
Each triangle $t_i$ is (uniquely) extendable to a $K_4=:K(i)$ in $R$.
Therefore, $t$ belongs to $K(1)\times\cdots\times K(k)$.

It follows that 
$$
\partial(R^k) \ge 8^k \partial(K)\ge2^{3k}2^{2k-1}=2^{5k}/2,
$$
yielding the desired bound~\refeq{SkRk}.
\end{proof}

\subsection{Domination number}

We conclude this section with a discussion of the domination number.
Recall that, by Estimate \refeq{gamma-int-gap}, the integrality gap of $\gamma_f$
over $n$-vertex graphs is bounded by $1+\ln n$, which is also an upper bound for
the invariance ratio of~$\gamma$. We now prove a lower bound that is tight up
to a constant factor.

A \emph{circulant graph} of $n$ vertices is a Cayley graph $\cay(\bZ_n,C)$
of the cyclic group $\bZ_n$. Here, the \emph{connection set} $C$ is a subset
of $\bZ_n$ such that $0\notin C$ and $C=-C$. Two vertices $x,y\in\bZ_n$
are adjacent in $\cay(\bZ_n,C)$ if $x-y\in C$.
Let $q$ be a prime power such that $q\equiv1\pmod4$.
The \emph{Paley graph} on $q$ vertices is the circulant graph
$\cay(\bZ_q,C_q)$ where $C_q$ is the multiplicative subgroup of $\bZ_q^*$
formed by all quadratic residues modulo~$q$.

Part 1 of Theorem \ref{thm:gamma} below implies that the integrality gap of the fractional
domination number for Paley graphs is logarithmic.
This was shown in \cite{ChappellGH17} for random graphs.
Problem 3.4 in \cite{ChappellGH17} aks whether a logarithmic gap
can be shown by an explicit construction. Our approach answers this
in the affirmative as Paley graphs are explicitly constructed.
It is actually not a big surprise that, using Paley graphs, one can replace 
a probabilistic argument in the situation like this. Indeed, it is
well known that, in some precise sense, the Paley graphs have the same first-order properties
as random graphs. Note in this respect that the property $\gamma(G)>k$ of a graph $G$
is clearly expressible in first-order logic. The proof also uses
the \WL1-invariance of the fractional domination number.

\begin{theorem}\label{thm:gamma}\hfill
  \begin{enumerate}[\bf 1.]
  \item 
Let $n$ be a prime power such that $n\equiv1\pmod4$, and $G_n$ denote the Paley graph on $n$ vertices.
Then $\gamma(G_n)\ge(\frac12-o(1))\log_2 n$ while $\gamma_f(G_n)\le2$.
  \item 
  For infinitely many $n$, there are \WL1-equivalent $n$-vertex
  graphs $G_n$ and $H_n$ such that $\gamma(G_n)/\gamma(H_n)\ge(\frac14-o(1))\log_2 n$.
  \end{enumerate}
\end{theorem}

  \begin{proof}
The \emph{$k$-extension property}  
says that, for every two disjoint sets of vertices $X$ and $Y$ with $|X\cup Y| \le k$
(where one of $X$ and $Y$ can be empty), 
there is a vertex $z\notin X\cup Y$ adjacent to each $x\in X$ and to no $y\in Y$.
It is known \cite{BlassEH81,BollobasT81} that the Paley graph $G_n$ has the $k$-extension 
property for each $k$ such that $n>k^22^{2k-1}$. The last condition is true, in particular,
for $k=k(n)$ where 
$$k(n)=\lceil\frac12\log_2n-\log_2\log_2n\rceil.$$
Note that the $k$-extension property implies that the domination number exceeds $k$.
It follows that 
$$\gamma(G_n)>k(n)\ge\frac12\log_2n-\log_2\log_2n.$$

Now, let $n=4s+1$ and consider the circulant graph $H_n=\cay(\bZ_n,\{\pm1,\ldots,\pm s\})$.
The set $\{0,2s\}$ is dominating in $H_n$, which implies $\gamma(H_n)=2$.
Since $G_n$ and $H_n$ are regular graphs of the same degree, $G_n\eqkkwl1 H_n$.
By Rubalcaba's result on the \WL1-invariance of the fractional domination number
(see Section \ref{ss:gamma_f}), we therefore have
$$\gamma_f(G_n)=\gamma_f(H_n)\le\gamma(H_n)=2,$$
which implies Part~1.

The pair of graphs $G_n$ and $H_n$ yields also a proof of Part~2.
  \end{proof}

\section{Conclusion}
We have studied Weisfeiler-Leman invariance of the fractional packing
number $\pi^F_f$ and its edge-disjoint variant $\rho^F_f$. As a
starting point of our analysis, we have shown that the fractional
matching number of a hypergraph is \WL1-invariant, where a hypergraph
(an instance of the Set Packing problem) is represented by its
incidence graph. For a pattern graph $F$ with $\ell$ vertices, this
already implies the \WL{(2\ell-1)}-invariance of $\pi^F_f$
(see Corollary \ref{cor:interpr}).  Our main result, Theorem
\ref{thm:htwpack}, is more precise. It shows for pattern graphs $F$ of
hereditary treewidth $k$ that $\pi^F_f$ is {\WL k}-invariant. This is
optimal in some cases, for example, when $F=K_3$ or $F=K_{1,s}$.  The
latter case, i.e., when $F$ is a star, includes the fractional
matching number and adds to the list of fractional graph parameters
invariant under fractional isomorphisms.

An important motivation for the study of $\mathrm{WL}$-invariance is
that this concept can be applied to showing lower bounds on the integrality gap of a
fractional graph parameter~$\pi_f$:
\begin{itemize}
\item 
we first prove that $\pi_f$ is {\WL k}-invariant for an integer~$k$;
\item 
then, we estimate the {\WL k}-invariance ratio of $\pi$ from below,
which provides us with a lower bound for the integrality gap
of~$\pi_f$.
\end{itemize}

Interestingly, this approach yields tight bounds in some cases like
for the fractional matching number $\nu_f$, the fractional cover
number $\tau_f$, the fractional domination number $\gamma_f$, and the
fractional edge-disjoint triangle packing number $\rho^{K_3}_f$.
While the first two examples are based on simple graph sequences\footnote{For $\nu_f$, 
it is enough to consider $G=C_{6s}$ and $H=2s\,C_3$, as
  discussed in Section \ref{s:intro}. For $\tau_f$, consider the pair
  of $\WL1$-equivalent graphs $G=2K_s\cup sK_2$ and $H=K_{s,s}$,
  where $G$ is the $2s$-vertex graph obtained from two vertex-disjoint
  $s$-cliques by adding a perfect matching between them, and note that
  $\tau(G)=2s-2$ while $\tau(H)=s$.}, the last two cases are
considered in Theorems \ref{thm:gamma} and \ref{thm:ratio}
respectively. It is worth mentioning that, for all these parameters,
the lower bounds for the invariance ratio are shown using explicit graphs
or graph sequences, which in the case of the fractional domination number
gives a solution of \cite[Problem~3.4]{ChappellGH17}.\footnote{%
We have just learned that a different solution is given in~\cite{Chappell}.}
In general, the connection between the integrality gap and invariance ratio is 
an intriguing question that merits further study. Are
these two characteristics close to each other for all WL invariant graph parameters?

Another important question, whose discussion is initiated in Section \ref{ss:inv-diff}, is whether the
\WL2-invariance of $\rho^{K_3}_f$ can be used for obtaining tight
bounds on the additive integrality gap of this parameter.

\subsection*{Acknowledgement.}
We thank Glenn Chappell for sending us his manuscript~\cite{Chappell}.

\end{document}